\RequirePackage{amsmath}

\documentclass[runningheads]{llncs}
\usepackage{comment}
\usepackage[latin1]{inputenc}
\usepackage{amssymb,amsfonts}
\usepackage{graphicx}
\usepackage{cleveref}
\usepackage{xspace}
\usepackage{subcaption}
\usepackage{mathtools}
\usepackage{xcolor, soul}
\definecolor{mycolor}{gray}{0.85}
\sethlcolor{mycolor}
\usepackage{natbib}
\usepackage{apalike}
\usepackage[margin=1.0in]{geometry}
\usepackage{todonotes}

\newtheorem{observation}[lemma]{Observation}

\providecommand{\keywords}[1]{\textbf{\text{Keywords:}} #1}

\definecolor{emphcol}{gray}{.75}
\newcommand{\take}[1]{\colorbox{emphcol}{$#1$}}

\renewcommand{\epsilon}{\varepsilon}

\newcommand{\PoA}{\textbf{PoA}\xspace}
\newcommand{\SPoA}{\textbf{SPoA}\xspace}
\newcommand{\PoS}{\textbf{PoS}\xspace}
\newcommand{\SPoS}{\textbf{SPoS}\xspace}
\newcommand{\aSPoS}{\textbf{adaptive SPoS}\xspace}

\newcommand{\instances}{\mathcal{P}_{n,m}}
\newcommand{\trees}{\mathcal{T}_{n,m}}
\newcommand{\perms}{\mathcal{S}_n}

\begin{document}

\title{Sequential Solutions in Machine Scheduling Games}

%



\date{}

\author{Cong Chen\inst{1} \and Paul Giessler\inst{2} \and Akaki Mamageishvili\inst{3} \and Mat\'u\v{s} Mihal\'ak\inst{2} \and Paolo Penna\inst{4}}
\institute{
School of Business Administration, South China University of Technology, Guangzhou, China \and
Department of Data Science and Knowledge Engineering, Maastricht University, the Netherlands \and
Department of Management, Technology and Economics, ETH Zurich, Switzerland \and
Department of Computer Science, ETH Zurich, Switzerland }
\maketitle

\begin{abstract}
	We consider the classical machine scheduling, where $n$ jobs need to be scheduled on $m$ machines, and where job $j$ scheduled on machine $i$ contributes $p_{i,j}\in \mathbb{R}$ to the load of machine $i$, with the goal of minimizing the makespan, i.e., the maximum load of any machine in the schedule.
	We study inefficiency of schedules that are obtained when jobs arrive sequentially one by one, and the jobs choose themselves the machine on which they will be scheduled, aiming at being scheduled on a machine with small load. 
	%
	%
	We measure the inefficiency of a schedule as the ratio of the makespan obtained in the worst-case equilibrium schedule, and of the optimum makespan. This ratio is known as the \emph{sequential price of anarchy (\SPoA)}. 
	We also introduce two alternative inefficiency measures, which allow for a favorable choice of the order in which the jobs make their decisions.
	%
	As our first result, we disprove the conjecture of~\cite{Hassin} claiming that the sequential price of anarchy for $m=2$ machines is at most 3. 
	%
	We show that the sequential price of anarchy grows at least linearly with the number $n$ of players, assuming arbitrary tie-breaking rules. That is, we show $\SPoA \in \Omega(n)$. Complementing this result, we show that $\SPoA \in O(n)$, reducing previously known exponential bound for $2$ machines.
	Furthermore, we show that there exists an order of the jobs, resulting in makespan that is at most linearly larger than the optimum makespan.
	To the end, we show that if an authority can change the order of the jobs adaptively to the decisions made by the jobs so far (but cannot influence the decisions of the jobs), then there exists an adaptive ordering in which the jobs end up in an optimum schedule.
	%
%
\end{abstract}


%
\section{Introduction}\label{sec:intro}
We consider the classical optimization problem of scheduling $n$ jobs on $m$ \emph{unrelated} machines. 
In this problem, each job has a (possibly different) processing time on each of the $m$ machines, and a schedule is simply an assignment of jobs to machines. For any such schedule, the load of a machine is the sum of all processing times of the jobs assigned to that machine. In this optimization problem, the objective is to find a schedule minimizing the \emph{makespan}, that is, the maximum load among the machines. 

In the \emph{game-theoretic} version of this scheduling problem, also known as the \emph{load balancing game}, jobs correspond to players who \emph{selfishly} choose the machine to which the job is assigned. 
The cost of a player is the \emph{load} of the machine to which the player assigned its own job.
Such a setting models, for example, the situation where the machines correspond to servers, and the communication with a server has a latency that depends on the total traffic to the server.

The decisions of the players lead to some \emph{equilibrium} in which no player has an incentive to deviate, though the resulting schedule may not necessarily be optimal in terms of makespan. Such an equilibrium might have a rather high \emph{social cost}, that is, the makespan of the corresponding schedule\footnote{When each player chooses deterministically one machine, this definition is obvious. When equilibria are \emph{mixed} or randomized, each player chooses one machine according to some probability distribution, and the social cost is the expected makespan of the resulting schedule.} is not guaranteed to be the optimal one, as in \Cref{ex:bad-Nash} below.

\begin{example}[two jobs on two unrelated machines \cite{Awerbuch+etal/2006}]\label{ex:bad-Nash}
	Consider two jobs and two unrelated machines, where the processing times are given by the following table:
	\[
	    \begin{array}{c|c|c|}
		                    & \text{job~1} & \text{job~2} \\ \hline
	    \text{machine~1}	& 1            & \take{\ell}  \\ \hline
	    \text{machine~2}	& \take{\ell}  & 1            \\ \hline
	    \end{array}
	 \]
	 The allocation represented by the gray box is a pure Nash equilibrium in the load balancing game (if a job moves to the other machine, its own cost increases from $\ell$ to $\ell + 1$), and has makespan $\ell$. The optimal makespan  is $1$ (swap the allocations). This example shows that the makespan of an equilibrium can be arbitrarily larger than the optimum.
\end{example}

The inefficiency of equilibria is a central concept in algorithmic game theory. Typically, one aims to quantify the \emph{efficiency loss} resulting from a \emph{selfish behavior} of the players, where the loss is measured in terms of the social cost.
Arguably, the two most popular measures of inefficiency of equilibria are the \emph{price of anarchy (\PoA)} \cite{PoAoriginal} and the \emph{price of stability (\PoS)} \cite{PoSoriginal}, which, intuitively, consider the \emph{most pessimistic} and the \emph{most optimistic} scenario:
\begin{itemize}
	\item The \emph{price of anarchy} is the ratio of the cost of the \emph{worst} equilibrium over the \emph{optimal social cost};
	\item The \emph{price of stability} is the ratio of the cost of the \emph{best} equilibrium over the \emph{optimal social cost}. 
\end{itemize}
The price of anarchy corresponds to the situation (anarchy) in which there is no authority, and players converge to some equilibrium by themselves. In the price of stability, one envisions that there are means to suggest the players how to play, and if that is an equilibrium, then they will indeed follow the suggestion, as no unilateral deviation can improve a player's individual cost. Furthermore, the price of stability provides a lower bound on the efficiency loss of an equilibrium outcome, if, for example, no equilibrium is actually a social optimum. 

Example~\ref{ex:bad-Nash} thus shows that the \textbf{price of anarchy} of load balancing games is \textbf{unbounded even for two jobs and two machines}. Interestingly, the price of stability instead is one (\PoS=1), for any number of jobs and any number of machines. This is because there is always an optimal solution that  is also a pure Nash equilibrium \cite{Even-Dar+Kesselman+Mansour/2003} (see  Section~\ref{sec:further-related-work} for details). In a pure Nash equilibrium, players choose their strategies deterministically, as opposed to \emph{mixed} Nash equilibria. In this work, we will also focus on the case in which players act deterministically, though in a sequential fashion (see below).

As the price of anarchy for unrelated machines is very high (unbounded in general),  one may ask whether Nash equilibria are really what happens as an outcome in the game, or whether a central authority, which cannot influence the choices of the players (jobs), may alter some aspects of the scheduling setting, and as a result, improve the performance of the resulting equilibria.


Motivated by these issues, 
in \cite{originalSPOA} the authors consider the variant in which players, instead of choosing their strategies simultaneously, play sequentially taking their decisions based on the previous choices and also knowing the order of players that will make play. Formally, this corresponds to an \emph{extensive-form game}, and the corresponding equilibrium concept is called a \emph{subgame-perfect equilibrium}. Players always choose their strategy deterministically. The resulting inefficiency measure is called the \emph{sequential price of anarchy (\SPoA)}.

There are two main motivations to study a sequential variant of the load balancing game. First, assuming that all players decide simultaneously to choose the machine to process their jobs is a too strong and unnatural modeling assumption in many situations; furthermore, expecting that all players choose the worst-case machine, as was the case in Example~\ref{ex:bad-Nash}, is unnatural as well. Second, one may have the power to explicitly ask the players to make sequential decisions, and make this the policy, which the players are aware of, with the view of lowering the loss of efficiency of the resulting equilibrium schedules. In a sense, such an approach of adjusting the way the players access the machines resembles \emph{coordination mechanisms} \cite{Christodoulou+Koutsoupias+Nanavati/2009}, which are scheduling policies aiming to achieve a small price of anarchy (see Section~\ref{sec:further-related-work} for more details).

\subsection{Prior results (\SPoA for unrelated machines)}\label{sec:prior:SPoA}
The first bounds on the \SPoA for unrelated machines have been obtained in \cite{originalSPOA}, showing that 
\[
n \leq \SPoA \leq m\cdot 2^n\text{.}
\]
Therefore, \SPoA is \emph{constant} for a constant number of machines and jobs, while \PoA is \emph{unbounded} even for two jobs and two machines (recall Example~\ref{ex:bad-Nash}). The large gap in the previous bound naturally suggests the question of  what happens for \emph{many jobs} and \emph{many machines}. This was addressed by \cite{bilo} which improved significantly the prior bounds by showing that 
\[
    2^{\Omega(\sqrt{n})} \leq  \SPoA \leq 2^n\text{.}
\]
At this point, one should note that these lower bounds use a \emph{non-constant} number of machines. In other words, it still might be possible that for a \emph{constant number of machines} the \SPoA is constant. For \emph{two machines}, \cite{Hassin} proved a lower bound $\SPoA \geq 3$, and in the same work the authors made the following conjecture: 
\begin{conjecture}\label{conjecture}\cite{Hassin}
	 \em For two unrelated machines, $\SPoA=3$ for any number of jobs. 
\end{conjecture}

\subsection{Our contributions}\label{sec:contributions}
In this paper, we disprove Conjecture~\ref{conjecture} by showing that in fact, \SPoA on two machines is \emph{not even constant}.
Indeed, it must grow linearly and the conjecture fails already for few jobs: 
\begin{itemize}
	\item For \emph{five jobs} we have $\SPoA \geq 4$ (Theorem~\ref{th:LB:five-jobs});
	\item In general, with arbitrary tie-breaking rules, it holds that $\SPoA \geq \Omega(n)$ (Theorem~\ref{counterexample}). 
\end{itemize}
Note that the result of Theorem~\ref{counterexample} uses suitable player-specific tie-breaking rules (see Definition~\ref{def:arbitrary-ties}). We discuss the implications of using tie-breaking rules more in detail at the end of this subsection. 

While Theorem~\ref{th:LB:five-jobs} settles the conjecture, the result of Theorem~\ref{counterexample}  says that \SPoA is non-constant already for two machines (as the number of jobs grows) for generic tie-breaking rules. We actually conjecture that there exist instances for which the $\SPoA$ is unbounded without having ties.
Moreover, it implies a \emph{strong separation} with the case of \emph{identical} machines, where $\SPoA \leq 2-\frac{1}{m}$, for any number $m$ of machines \cite{Hassin}. 
In Theorem~\ref{thm:UB4SPoA} we show that \SPoA is upper bounded by $2(n-1)$, reducing the exponential upper bound obtained in~\cite{bilo} for arbitrarily many machines to linear bound for $2$ machines.
 
The original idea behind the notion of price of stability (\PoS) is that an authority can suggest to the players how to play:
\begin{quote}
    \emph{[...] The best Nash equilibrium solution has a natural meaning of stability in this context -- it is the optimal solution that can be proposed from which no user will defect. [...] As a result, the global performance of the	system may not be as good as in a case where a central authority can simply dictate a solution; rather, we need to understand the quality of solutions that are consistent with self-interested behavior.} \cite{PoSoriginal}
 \end{quote}
 We borrow this idea of  an authority suggesting desirable equilibria.
 Specifically for our setting, the authority suggests the order in which players make their decisions, so to induce a good equilibrium. This can be viewed as the price of stability (\PoS) for these sequential games. We introduce this notion in two variants (a weaker and a stronger):
\begin{itemize}
	\item \emph{Sequential Price of Stability (\SPoS).} The authority can choose the order of the players' moves. This order determines the tree structure of the corresponding game. 
	\item \emph{Adaptive Sequential Price of Stability (\aSPoS).} The authority decides the order of the players' moves \emph{adaptively} according to the choices made at each step. 
\end{itemize}
The study of these two notions for two unrelated machines is also motivated by our lower bound, and by the lack of any good upper bound on this problem. We prove the following upper bounds for two unrelated machines (\Cref{th:SPoS,optimum_sequence}):
\begin{align*}
	\SPoS \leq \frac{n}{2}+1 \ , && \aSPoS = 1 \ .
\end{align*}
The next natural question is to consider \emph{three} or more machines. Here we show an impossibility result, namely $\aSPoS \geq 3/2$ already for three machines (\Cref{th:adaptive:LB}). That is, even with the strongest type of adaptive authority, it is not possible to achieve the optimum. This shows a possible  disadvantage of having players capable of complex reasoning, like in extensive-form games. In the classical strategic-games setting, where we consider pure Nash equilibrium, here is an optimum which \emph{is} an equilibrium, that is, $\PoS=1$ for any number of machines and jobs. This result follows from \cite{Even-Dar+Kesselman+Mansour/2003} (see Section~\ref{sec:further-related-work} for details)

As mentioned above, some of our results rely on the use of a suitable tie-breaking rules. Using tie-breaking rules  to prove lower bounds on the \SPoA is not new:  in \cite{SPOAcongestion} the authors showed that, in \emph{routing games}, the sequential price of anarchy is \emph{unbounded}. Their proof is based on  carefully chosen tie-breaking rules.  This way of using tie-breaking rules is not part of the players' strategy interactions. In contrast, some works consider settings where among equivalent choices, each player $i$ can use the one that hurts prior agents who
chose a strategy that player $i$ would prefer they had not chosen (see \cite{tie-breaking}).

\subsection{Further related work}\label{sec:further-related-work}
The load balancing games considered in this work are one of the most studied models in algorithmic game theory (see, e.g., \cite{PoAoriginal,koutsoupias2003approximate,fiat2007strong,andelman2009strong,Czumaj+Vocking/2007,feldmann2003nashification,Epstein2010}).  In all these works, players correspond to jobs, their cost is the load of the machine they choose, and the social cost is defined as the makespan of the jobs allocation.  In particular, the seminal paper \cite{PoAoriginal} which introduced the concept of the price of anarchy,   considers the case of \emph{identical} and \emph{related} machines,   two simpler versions of unrelated machines (related machines is the setting where each machine has a speed, each job has a certain size, and the processing time equals the job size divided by the machine speed; the case of identical machines is the restriction in which all speeds are the same). 

Interestingly,  the price of anarchy for \emph{related} or \emph{identical} machines is much better than in the case of unrelated machines (where the price of anarchy is unbounded). Indeed, for related and identical machines, the price of anarchy is \emph{bounded} for any \emph{constant number of machines}  \cite{PoAoriginal,Czumaj+Vocking/2007,koutsoupias2003approximate,Finn+Horowitz/1979,fiat2007strong,feldmann2003nashification,Epstein2010} (some of these results give bounds also for \emph{mixed} Nash equilibria). Specifically,  for pure Nash equilibria,   $\PoA=(2-\frac{2}{m+1})$ for identical machines as implied by the analysis of  \cite{Finn+Horowitz/1979}, while $\PoA = O(\frac{\log m}{\log \log m})$  for related machines \cite{Czumaj+Vocking/2007}. 

%


As already mentioned above, the \PoS
for \emph{unrelated} machines is $1$. This is due to the work  \cite{Even-Dar+Kesselman+Mansour/2003} which shows that, starting from any schedule, an iterative process of applying unilateral improving-strategy changes of players leads to a pure Nash equilibrium (the same property has been observed earlier in \cite{Fotakis+etal} for related machines). This condition implies the existence of a pure Nash equilibrium.

Load balancing games on identical and related machines are a special case of \emph{weighted singleton congestion games}. In a singleton weighted congestion game, there are $m$ resources, and $n$ players, each player $i$ having a weight $w_i$.
Every resource $r$ has a cost function $c_r$ associated with it. In the game, every player $i$ chooses one resource $s_i$ as its strategy, resulting in cost $c_r(\sum_{i:s_i=r} w_i)$ of the resource, which is also the cost of every player $i$ that chooses resource $r$ as its strategy. Obviously, seeing the machines in the load balancing games as resources, seeing the jobs as the players, seeing the job sizes as weights $w_i$, and setting $c_r(x)=x/\text{speed}_r$, the singleton congestion game models the load balancing games on $m$ related machines, where machine $r$ has speed $\text{speed}_r$, and the processing time of job $i$ on machine $r$ is $w_i/\text{speed}_r$.
Load balancing games on unrelated machines have, to the best of our knowledge, no counterpart in congestion games.

Requiring that players make their decisions sequentially, according to a given and known order can be seen as a mean of a central authority that can control  access to the resources (machines), but not the choices of the players (jobs).
In this sense, changing the access from simultaneous to sequential can be seen as a kind of control mechanism like a \emph{coordination mechanism} \cite{Christodoulou+Koutsoupias+Nanavati/2009}. In load balancing games where the cost of a player (job) is the completion time of the job (and not the total load of the machine on which the job is scheduled), a coordination mechanism is a scheduling policy, one for every machine, which determines the order of the jobs in which they will be scheduled on the machine. The scheduling policy needs to be fixed and (publicly) known to the players.
For load balancing games in normal form (i.e., where players make simultaneous decisions, as opposed to the sequential decisions, which we consider in this paper), coordination mechanisms have been studied both for the version where the social cost is the makespan (see, e.g., \cite{Immorlica+Li+Mirrokni+Schulz/2009,Caragiannis/2013,Azar+Fleischer+Jain+Mirrokni+Svitkina/2015,Caragiannis+Fanelli/2016} and the references therein), or the total (weighted) completion time (see, e.g., \cite{Hoeksma+Uetz/2012,Correa+Queyranne/2012,Cole+etal/2015,Abed+Correa+Huang/2014,Zhang+Zhang_Du+Bai/2018} and the references therein).

As already discussed above, the concept of a sequential price of anarchy is not new. In addition to the results for unrelated machines discussed in Section~\ref{sec:prior:SPoA},  the sequential price of anarchy has been studied also for other games. These include congestion games with affine delay functions  \cite{SPOAcongestion}, isolation games \cite{Angelucci2015}, and network congestion games  \cite{correa2015curse}. Interestingly, the latter work shows that the sequential price of anarchy for these games is \emph{unbounded}, as opposed to the price of anarchy which was known to be $5/2$.

Naturally, there is a huge  literature on the classical  algorithm-theoretic research on machine scheduling, see, e.g., the textbook \cite{Pinedo/2012} and the survey \cite{Chen+Potts+Woeginger/1999} for fundamental results and further references.

\section{Preliminaries}\label{sec:preli}
In unrelated machine scheduling there are $n$ jobs and $m$ machines, and the processing time of job $j$ on machine $i$ is denoted by $p_{ij}$. A solution (or schedule) consists of an assignment of each job to one of the machines, that is, a vector $s=(s_1,\ldots,s_n)$ where $s_j$ is the machine to which job $j$ is assigned to. 
The \emph{load} $l_i(s)$ of a machine $i$ in schedule $s$ is the sum of the processing times of all jobs allocated to it, that is, $l_i(s) = \sum_{j: s_j=i}p_{ij}$. The social cost of a solution $s$ is the \emph{makespan}, that is, the maximum load among all machines.  

Each job $j$ is a \emph{player} who attempts to minimize her own cost $cost_j(s)$, that is, the load of the machine she chooses: $cost_j(s) = l_{s_j}$. Every player $j$ decides $s_j$, the assignment of job $j$ to a machine. The combination of all players strategies gives a schedule $s=(s_1,\ldots,s_n)$.

In the extensive-form version of these games, players play sequentially; they decide their strategies based on the choices of the previous players and knowing that the remaining players will play rationally. 
We consider a \emph{full information} game. As players enter the game sequentially, they can compute their optimal moves by the so-called \emph{backward induction}: the last player makes her move greedily, the player before the last makes the move also greedily (taking into account what the last player will do), and so on. Any game of this type can be modeled by a \emph{decision tree}, which is a rooted tree where the non-leaf vertices correspond to the players in certain states, while edges correspond to the strategies available to the players in a given state.

Each leaf corresponds to a solution (schedule), which is simply the strategies on the unique leaf-to-root path. Given the processing times $P=(p_{ij})$, the players can compute the loads on the machines in each of the leaves. In  case of ties, all players know the deterministic tie-breaking rules of all the other players. A player can calculate what the final outcome would be for each of her strategies, and choose the strategy that minimizes her cost. This method is called backward induction. Strategies obtained in this way for each internal node constitute what is called the \emph{subgame-perfect equilibrium}: for each subtree, we know what is the outcome achieved by the players in this subtree if they play rationally.  We usually represent the strategies (edges) that are chosen by players in the \textbf{subgame perfect equilibrium} in  \textbf{bold}, and the other strategies as \emph{dashed} edges. 

It is easy to see that a subgame-perfect equilibrium always exists and it is unique, for given tie-breaking rules. On the other hand, its computation is difficult, as proved in~\cite{originalSPOA}:

\begin{theorem}~\cite{originalSPOA}
Computing the outcome of a subgame perfect equilibrium in Unrelated Machine Scheduling is PSPACE-complete.
\end{theorem}

\paragraph{Notation and formal definitions.}
We consider $n$ jobs and $m$ machines, denoted by $J=(J_1,J_2,\ldots,J_n)$ and $M=(M_1,M_2,\ldots,M_m)$ respectively. The processing times are given by a matrix $P=(p_{ij})$, with $p_{ij}$ being the processing time of job $J_j$ on machine $M_i$. The set of all such nonnegative $n \times m$ matrices is denoted by $\instances$ and it represents the possible instances of the game. For any $P\in \instances$ as above, we denote by $\trees$ the set of all possible depth-$n$, complete $m$-ary  decision trees where each path from the root to a leaf contains every job (player) exactly once. The whole game (and the resulting subgame perfect equilibrium) is fully specified by $P$, $T$, and the \emph{tie-breaking rule} used by the players. The most general -- worst case -- scenario is that ties are arbitrary (see Definition~\ref{def:arbitrary-ties}). In the following, we do not specify the dependency on the ties, and simply denote by  $SPE(P,T)$ the cost (makespan) of the subgame perfect equilibrium of the game. 
One type of worst-case analysis is to assume the players' order to be adversarial, and the tree $T$ being chosen accordingly. This is the same as saying that players arrive in a fixed order (say $J_1,J_2,\ldots, J_n$) and their costs $P$ is chosen in an adversarial fashion. 
In this case, we simply write $SPE(P)$ as the tree structure is fixed.
For a fixed order $\sigma$  (a permutation) of the players, and costs $P$, we also write $SPE(P,\sigma)$ to denote the quantity $SPE(P,T)$ where $T$ is the tree resulting from this order $\sigma$ of the players.
The optimal social cost (makespan) is denoted by $OPT(P)$.

We next introduce formal definitions to quantify the inefficiency of subgame perfect equilibria in various scenarios (from the most pessimistic to the most optimistic).  
The \emph{sequential price of anarchy (\SPoA)} compares the worst subgame perfect equilibrium with the optimal social cost,  $$\text{\SPoA} = \sup_{P \in \instances} \frac{SPE(P)}{OPT(P)}.$$
	In the \emph{sequential price of stability (\SPoS)}, we can choose the order $\sigma$ in which players play depending on the instance $P$. The resulting subgame perfect equilibrium has cost $SPE(P,\sigma)$, which is then compared to the optimum, $$\SPoS = \adjustlimits \sup_{P\in \instances}\min_{\sigma\in \perms}\frac{SPE(P,\sigma)}{OPT(P)},$$ where  $\sigma$ ranges over all permutations $\perms$ of the $n$ players. 
	In \emph{adaptive sequential price of stability (\aSPoS)}, we can choose the whole structure of the tree, meaning that for each choice of a player, we can adaptively choose which player will play next. This means that every path from any leaf to the root corresponds to a permutation of the players. The adaptive price of stability is then defined as $$\aSPoS = \adjustlimits \sup_{P\in \instances}\min_{T\in \trees}\frac{SPE(P,T)}{OPT(P)}.$$ 
Note that by definition $\aSPoS \leq \SPoS \leq \SPoA$. 

\section{Linear lower bound for \SPoA}\label{sec:counterexample}

In this section, we consider the sequential price of anarchy for \emph{two} unrelated machines. In \cite{Hassin} the authors proved a lower bound $\SPoA\geq 3$ for this case, and they conjectured that this was also a tight bound. We show that unfortunately this is not the case: Already for five jobs, $\SPoA \geq 4$, and with more jobs the lower bound grows linearly, i.e., $\SPoA = \Omega(n)$.

\subsection{A lower bound for $n=5$ players}

\begin{theorem}\label{th:LB:five-jobs}
	For two machines and at least five jobs, the \SPoA is at least $4$.
\end{theorem}
The proof is in the appendix.

\subsection{Faster linear program formulation}
	Our first lower bound for $n=5$ players has been obtained by solving a linear program, suggested in \cite{Hassin}. 
	A crucial achievement for the speedup of the program is the discovery of
	the property that we describe in the following. This property allows excluding a large
	number of combinations for the last layer of the tree. As the last layer of the
	tree represents more than half of the internal nodes in the tree, the
	number of combinations that have to be generated can be reduced drastically.
	For example, for $n=5$, the improvement is from $~2\cdot 10^9$ to $~6\cdot 10^6$.

In this approach of linear programming, the variables are the processing times $\{p_{ij}, 1\leq i \leq m, 1\leq j \leq n \}$, and the approach essentially goes  as follows: 
\begin{itemize}
	\item Fix the subgame perfect equilibrium structure, that is, the sequence of players and all the decisions in the internal nodes, this also gives the sequential equilibrium;
	\item Fix the leaf which is the optimum, and impose that the optimum makespan is at most $1$; 
	\end{itemize}
	For every fixed  subgame perfect equilibrium tree structure, we have one constraint for each internal node (decision of a player).
	The optimum state (leaf) should also be fixed and both numbers have to be assumed to be at most $1$ by adding two additional constraints to the linear program. By maximizing the maximum value of loads on the machines in the leaf which corresponds to the sequential equilibrium, we get the worst case example for this particular tree structure. There are $2^{n}-1$ internal nodes in the decision tree with $n$ players. Therefore, this  approach requires exploring $2^{2^{n}-1}$ many possible subgame perfect equilibria tree structures, and for each of them, we have to decide where is the optimum among $2^n$ leaves and solve a linear program of size $2^n \times O(n)$.
	
	We managed to solve the case $n=5$ players with the aid of a computer program which explored all possible  tree structure leading to  subgame perfect equilibria; this has been achieved by understanding the structure and by breaking certain symmetries to reduce the search space, as we explain next. 
	It is clear that the extremely fast growing number of possible tree structures
	makes the program very time-consuming even for small values. Consequently, we tried to exclude combinations from the
	computation, i.e. we avoid starting the linear programming solver for
	certain tree structures. There are some trivial cases that we present for the intuition. The first is that not
	all leaves of the tree should be tested for the position of the optimum. Both
	the leftmost and the rightmost leaf nodes can be excluded from the possible optimum position due to
	the property that \SPoA is $1$ in both cases. For the same reason, all tree structures
	where the equilibrium is located at those extreme leaves can also be ignored.
	Additionally, the leaf where the equilibrium is achieved should be avoided for
	the optimum position.
	The next idea is that trees that are mirror images of other trees with regard
	to the vertical axis will lead to the same \SPoA.

	During the experimental investigations of possible outcomes based on the structure
	of the game tree, we found out that a relatively big part of game tree structures
	always leads to an infeasible linear program, regardless of the position of
	the optimum. Consider the  simple tree structure for \emph{two} players and \emph{two} machines, depicted in \Cref{fig:conflicting}. Solid lines represent the best responses in each node.
	It is obvious that, if the best response of player $2$ in the left node is to choose machine $1$, then the best response in the right node cannot be to choose  machine $2$.
	We generalized this observation to any number of players $n$ and a number of machines $m$, and apply it to the second lowest level of the tree, i.e., to the best responses of the last player $J_n$. 
	
	\paragraph{Additional notation.} We denote the nodes of a $T\in \trees$ by $H$ and the nodes in the level $i$ by $H_i$. Clearly, in the case of adaptive sequences, $H_i$ may contain different players, while in the case of fixed sequences, in $H_i$ all the nodes correspond to exactly one player. 
	
	Let $p(h)$ denote the parent node of node $h$ in tree $t$. We define the child index $c(h) = i \in M$ of node $h$ if $h$ is the $i$-th child of it parent $p(h)$, this means that the edge from $p(h)$ to $h$ represents that the agent corresponding to $p(h)$ selects $M_i$. Every node $h$ on the $j$-th layer of $t$ is defined by the choices of the agents playing before agent $J_j$. These decisions create a unique path from the root of $t$ to $h$. The nodes on the path from the root $r$ to $h$ are $P(h) = \{r,...,p(p(h)), p(h), h\} = \{h'\in H|h \in t_{h'}\}$ where $t_{h'}$ is the subtree of $t$ rooted at $h'$. These definitions allow us to define the set of agents that selected a certain machine $i$ before node $h$: $B(i,h) = \{ j\in J | \ p(h')\in H_j  \land h'\in P(h) \land  c(h') =i\}$. Then the observation is the following: 
	
	\begin{observation}
		For every pair of nodes in the second lowest layer $h,h'\in H_n$ if the best response of player $J_n$ in $h$ is $i$, then it implies that the best response in $h'$ is also $i$ if $\forall i' \in M\setminus i: B(i',h)\subseteq B(i',h')$ holds. 
		\end{observation} 
		
		\begin{figure}
			\centering
			\includegraphics[scale=.7]{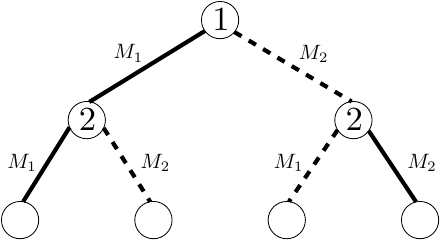}
			\caption{A tree for  the case of two players and two machines. The  responses of the players (bold edges) cannot be a subgame perfect equilibrium.}
			\label{fig:conflicting}
			\end{figure}

	\begin{remark}
	We could not find any example that would give a lower bound on sequential price of anarchy  (even locally) better than the one in the proof of \Cref{counterexample}. For $n\leq 7$, our computer program searched the whole space and the results obtained above are 			the best. We believe that the construction from the proof of the theorem gives the best possible lower bound example. 
	\end{remark}

\subsection{Linear lower bound}
Extending  the construction for $n=5$ players is non-trivial as this seems to require rather involved constants that multiply the $\epsilon$ terms. However, we notice that these terms only help to induce more involved tie-breaking rules of the following form:

\begin{definition}[arbitrary tie-breaking rule]\label{def:arbitrary-ties}
	We say that the tie-breaking rule is arbitrary if each player uses a tie-breaking rule between machines which possibly depends  on the allocation of all players.
\end{definition}


The following theorem gives our general lower bound:

\begin{theorem}
\label{counterexample}
Even for two machines, the \SPoA is at least linear in the number $n$ of jobs, in the case of arbitrary tie-breaking rule.
\end{theorem}

\begin{proof}
We consider the following instance with $n=3k-1$ jobs arriving in this order (from left to right),
%
\[\begin{array}{c||c|c|c|c|c|c|c|c|c|c|c|c|}
& J_1 & J_2 & J_3 & J_4 & J_5 & J_6 & \cdots & J_{3k-5} & J_{3k-4} & J_{3k-3} & J_{3k-2} & J_{3k-1} \\ \hline\hline
M_1 & \take{k+1} & 0 & \take{0} & k & \take{0} & \take{0} & \cdots & 3 & \take{0} & \take{0} & 1 & 2 \\ \hline
M_2 & 0 & \take{k} & k & \take{0} & k-1 & k-1 & \cdots & \take{0} & 2 & 2 & \take{1} &  \take{1} \\ \hline
\end{array}
\]
and show that the subgame perfect equilibrium is the gray allocation whose makespan is $k+2$, while the optimal makespan is $1$. This requires players to use the following tie-breaking rules in the first part: if player $J_1$ chooses machine $M_1$, then $J_2$ and $J_3$ prefer to avoid player $J_1$, that is, they choose the other machine in case of ties in their final cost.

We prove the claim above by induction on $k$. The base case is $k=2$ which follows directly from the example in \Cref{eq:SPE:five-players}, where we replace  $\epsilon>0$ with an equivalent tie-breaking rule (and set $\epsilon=0$). As for the inductive step, the proof consists of the following  steps (which we prove below):
\begin{enumerate}
    \item \label{itm:claim1}
    If the first three jobs choose their zero-cost machines, then all subsequent jobs implement the subgame perfect equilibrium on the same instance with $k'=k-1$. The cost of $J_1$, in this case, is $k'+2=k+1$.
    \item \label{itm:claim2} If the first job $J_1$ chooses $M_2$, then both $J_2$ and $J_3$ choose $M_1$.
    \item \label{itm:claim3} If the first job $J_1$ chooses $M_1$, then all subsequent players will choose the gray allocation (and therefore, the cost of $J_1$ is $k+1$ in this scenario as well). 
\end{enumerate}
The first two steps above imply that, if $J_1$ chooses machine $M_2$, then her cost  is $k+1$. Step~\ref{itm:claim3} says that the same cost occurs if $J_1$ chooses $M_1$. We assume the tie-breaking rule for player $J_1$, in this case, is that she prefers the cost $k+1$ on the first machine $M_1$. Therefore, by Step~\ref{itm:claim3}, $J_1$ will choose $M_1$ and all players choose the gray allocation in the subgame perfect equilibrium. Note that the cost on machine $M_1$ and $M_2$ is $k+1$ and $k+2$, respectively.

Next, we prove the three steps above:

\begin{proof}[of Step~\ref{itm:claim1}] 
	Note that the sequence starting from $J_4$ is the same sequence for $k' =k-1$. Since the first three jobs did not put any cost on the machines, we can apply the inductive hypothesis and assume that all subsequent players play the subgame perfect equilibrium. The resulting cost on machine $M_2$ will be $k'+2=k+1$, and this is the machine chosen by $J_1$. \qed
\end{proof}

\begin{proof}[of Step~\ref{itm:claim2}] 
	Choosing $M_2$ costs $J_2$ and $J_3$ at least $k$, no matter what the subsequent players do. If they instead choose $M_1$, by the previous claim, their cost is $k'+1=k$ which they both prefer given their tie-breaking rule. \qed
\end{proof}

%
%
%



\begin{proof}[of Step~\ref{itm:claim3}]
In this case, where $J_1$ is on machine $M_1$, we assume different tie-breaking rules for the last two players $J_{3k - 2}$ and $J_{3k - 1}$, depending on which of the two players $J_2$ and $J_3$ choose machine $M_2$.

{\bf Case 1:} player $J_2$  chooses machine $M_1$. In this case we assume that  player $J_3$ breaks ties in favor of $M_2$: we will show that choosing  $M_2$ results in the cost of $k+1$ in the end, instead of some cost on machine $M_1$, which we already know is at least $k+1$, because $J_1$ is already on machine $M_1$. If job $J_3$ gets assigned to machine $M_2$, then by backward induction we can show the following claim:

\begin{claim}
    No player among $J_4,\ldots, J_{3k-3}$ gets assigned to machine where she has non-zero cost.
\end{claim} 

\begin{proof}[of Claim]
Suppose none of them joins the non-zero cost machine. Then the last two players can add at most $3$ to the cost $k+1$ on the first machine and at most cost $2$ to the cost $k$ on the second machine. On the other hand, any job among $J_4, \ldots, J_{3k-3}$ adds at least that cost to the non-zero cost machine she joins. By backward induction we can assume that none of them chooses non-zero cost machine. 
\qed\end{proof}

Because of the last claim, only the last two players are left to decide. We assume that player $J_{3k-2}$ prefers to choose machine $M_1$ and pays the cost $k+2$ instead of choosing machine $M_2$ and paying the cost $k+2$ on that machine, while the last player $J_{3k-1}$ prefers to choose machine $M_2$ and incur the cost $k+1$. Therefore, job $J_2$ pays the cost $k+2$, but in this case we can assume that she prefers cost $k+2$ that she has to pay on machine $M_2$, given this is achievable. 

{\bf Case 2:} job $J_2$  gets assigned to machine $M_2$, as in the claimed subgame perfect equilibrium state. Similarly to the previous case, by backward induction we conclude that all players $J_3,\ldots, J_{3k-3}$ get assigned to the machine where they have cost $0$, and the last two jobs $J_{3k-2}$ and $J_{3k-1}$ choose machine $M_2$, here again we assume that player $J_{3k-2}$ prefers to pay $k+2$ on machine $M_2$ than to pay the same cost on machine $M_1$. In this way the cost on machine $M_1$ is $k+1$, while the cost on machine $M_2$ is $k+2$.

This finishes the proof of Step~\ref{itm:claim3} and of the theorem.  \qed

\end{proof}

\end{proof}

Note that it is important to have seemingly equivalent jobs $J_2$ and $J_3$. They use different tie-breaking rules, which creates the asymmetry between them and increases the \SPoA. 


We solved linear programs with strict inequalities obtained from the subgame perfect equilibria tree structure given in the example from the proof of  \Cref{counterexample}, by introducing small $\epsilon$ for strict inequalities. We found solutions for $n=8$ and $n=11$, that is linear programs are feasible. Therefore, at least for small $n$'s we can drop the assumption about tie-breaking rules. As the solutions replace the $\epsilon$ terms by rather more complicated coefficients,  we do not present them here. For the general case, we conjecture that the statement of \Cref{counterexample} holds without the assumption on the tie-breaking rules, and that the latter are merely used to make the analysis easier: 

\begin{conjecture}\label{conj:SPOA}
     For two machines, the \SPoA is at least linear in $n$.
\end{conjecture}



\section{Linear upper bound for \SPoA{}}
\label{sec:UP_for_spoa}
\paragraph{Additional notation.}
To prove the upper bound for \SPoA{}, we introduce some additional notation.
We define a vector $D=(d_1,d_2)$ of initial load on the machines before the jobs play the game.
Consequentially, the load of each machine $i$ becomes
\[
	l_i(D,s) = d_i + \sum_{j: s_j=i}p_{ij} \,,
\]
where $s=(s_1,s_2,\ldots,s_n)$ is the schedule (SPE) achieved by the jobs playing the game with initial load $D$ on the machines;
the cost of each job $j$ is 
\[
	cost_j(D,s) = l_{s_j}(D,s) \,.
\]
The notation for the makespan is renewed as $SPE_D(P)$ for the SPE with initial load $D$.
Additionally, we define $\Delta SPE(P)$ as the maximum possible increase of the makespan due to the players, with processing time $P$, for any initial load $D$:
\[
	{\Delta SPE(P)} = \sup_D\left\{ SPE_D(P) - ||D||_\infty \right\} \,.
\]
Moreover, for a given $P$, we use $P_{[u:v]}$ to represent the processing times only for jobs $(J_u,J_{u+1},\ldots,J_v)$, that is, $P_{[u:v]} = (p_{ij})$ where $j=u,u+1,\ldots,v$.

We first prove a key lemma showing that each job can only contribute a certain amount (bounded by the total minimum processing time) to the makespan:
\begin{lemma}\label{lem:2unrelated.key}
	$\Delta SPE \left( P_{[\ell:n]} \right)  - \Delta SPE \left( P_{[\ell+1:n]} \right)  \le \sum_{j = \ell}^n \min_i p_{ij}$ for $\ell = 1,2,\ldots,n-1$.
\end{lemma}
\begin{proof}
	For an arbitrary $\ell \in \{1,2,\ldots,n-1\}$, giving a processing time $P_{[\ell:n]}$ and an initial load $D$, we suppose w.l.o.g. that job $\ell$ chooses machine $1$.
	After job $\ell$ makes its decision, choosing machine 1, the game consists of the rest players, $\{\ell +1 ,\ell +2, \ldots,n\}$ starting with a new initial load $D' = (d_1 + p_{1\ell}, d_2)$.
	Thus,
	\[
		SPE_{D} \left( P_{[\ell:n]} \right) = SPE_{D'} \left( P_{[\ell+1:n]} \right)  \,.
	\]

	We first discuss the trivial case when $d_1 + p_{1\ell} < d_2$.
	In this case, it holds that $||D'||_\infty = ||D||_\infty = d_2$, which indicates that 
	\begin{equation*}
	SPE_{D} \left( P_{[\ell:n]} \right) =   SPE_{D'} \left( P_{[\ell+1:n]} \right) 
										\le \Delta SPE \left( P_{[\ell+1:n]} \right) + ||D'||_\infty 
										=	\Delta SPE \left( P_{[\ell+1:n]} \right) + ||D||_\infty \,,
	\end{equation*}
	that is,
	\begin{equation}\label{eq:2unrelated.case1}
		SPE_{D} \left( P_{[\ell:n]} \right) - ||D||_\infty \le \Delta SPE \left( P_{[\ell+1:n]} \right) \,.
	\end{equation}

	We then consider the other case $d_1 + p_{1\ell} \ge d_2$, that is, $||D'||_\infty = d_1 + p_{1,\ell}$.
	Let $s' = (s_{\ell+1},s_{\ell+2},\ldots,s_{n})$ be the schedule of the jobs $\{J_{\ell +1} ,J_{\ell +2}, \ldots,J_n\}$ playing with initial load $D'$.
	We know that
	\begin{equation}\label{eq:2unrelated.main}
		SPE_{D} \left( P_{[\ell:n]} \right) =   SPE_{D'} \left( P_{[\ell+1:n]} \right) = \max \{ l_1(D',s'), l_2(D',s') \} \,.
	\end{equation}
	\begin{claim}\label{claim:ub2unrelated}
		$\max \{ l_1(D',s'), l_2(D',s') \} \le l_1(D',s') + \sum_{j = \ell+1}^n \min_i p_{ij}$.
	\end{claim}
	\begin{figure}[tb]
		\centering
		\includegraphics[width=0.5\textwidth]{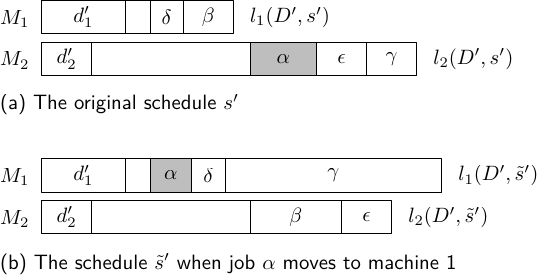}
		\caption{Proof of Claim}
		\label{fig:2unrelated}
	\end{figure}
	\begin{proof}[of Claim]
		If $\max \{ l_1(D',s'), l_2(D',s') \} = l_1(D',s')$, the claim is obviously true.
		Thus we only need to prove $l_2(D',s') \le l_1(D',s') + \sum_{j = \ell+1}^n \min_i p_{ij}$.

		Let $J_{\alpha}$ denote the last job who chooses machine 2 and has a longer processing time on machine 2 than on machine 1, i.e., $p_{1\alpha} \le p_{2\alpha}$.
		Let $\tilde{s}'$ be the new schedule if job $J_{\alpha}$ chooses machine 1.
		In schedule $\tilde{s}'$, the decisions of jobs $\{J_{\alpha+1},J_{\alpha+2},\ldots,J_n\}$ may be different from schedule $s'$.
		We divide the jobs $\{J_{\alpha+1},J_{\alpha+2},\ldots,J_n\}$ into $4$ subsets, namely, $\beta$, $\gamma$, $\delta$ and $\epsilon$, depending on the differences between $s'$ and $\tilde{s}'$ (as shown in Figure~\ref{fig:2unrelated}):
		\begin{itemize}
			\item Jobs in $\beta$ are on machine 1 in $s'$ but on machine 2 in $\tilde{s}'$;
			\item Jobs in $\gamma$ are on machine 2 in $s'$ but on machine 1 in $\tilde{s}'$;
			\item Jobs in $\delta$ are on machine 1 in both $s'$ and $\tilde{s}'$;
			\item Jobs in $\epsilon$ are on machine 2 in both $s'$ and $\tilde{s}'$.
		\end{itemize}

		For simplicity, we define some notations to represent the total processing time of the job sets:
		\begin{gather*}
		\alpha_1 = p_{1\alpha} ~,\qquad \quad 
		\beta_1 = \sum_{j \in \beta} p_{1j} ~,\qquad \quad
		\gamma_1 = \sum_{j \in \gamma} p_{1j} ~,\qquad \quad
		\epsilon_1 = \sum_{j \in \epsilon} p_{1j} ~, \\
		\alpha_2 = p_{2\alpha} ~,\qquad \quad
		\beta_2 = \sum_{j \in \beta} p_{2j} ~,\qquad \quad
		\gamma_2 = \sum_{j \in \gamma} p_{2j} ~,\qquad \quad
		\epsilon_2 = \sum_{j \in \epsilon} p_{2j} ~.
		\end{gather*}
		
		In the following, we will prove $l_2(D',s') \le l_1(D',s') + \alpha_1 + \gamma_2 + \epsilon_2$.
		According to the definition of job $J_{\alpha}$, we know the jobs in $\gamma$ and $\epsilon$ have shorter processing times on machine 2, thus it follows that $\alpha_1 + \gamma_2 + \epsilon_2 \le \sum_{j = \alpha}^n \min_i p_{ij} \le \sum_{j = \ell+1}^n \min_i p_{ij}$, meaning that the claim is true.

		We prove the inequality by contradiction, assuming that
		\begin{equation}\label{eq:2unrelated.contradiction}
			l_1(D',s') < l_2(D',s') - \alpha_1  - \gamma_2 - \epsilon_2 \,.
		\end{equation}
		Intuitively, when job $J_{\alpha}$ moves to machine 1, if the following jobs $\{J_{\alpha+1},J_{\alpha+2},\ldots,J_n\}$ make the same decisions as in schedule $s'$, the cost of job $J_{\alpha}$ (i.e., $l_1(D',s') + \alpha_1$) is lower than the cost (i.e., $l_2(D',s')$) of job $J_{\alpha}$ in the original schedule $s'$, since inequality \eqref{eq:2unrelated.contradiction} holds.
		To guarantee that job $J_{\alpha}$ has no incentive to move to machine 1, the cost of $J_{\alpha}$ when moving to machine 1 should be higher than $l_2(D',s')$.
		In other words, when job $J_{\alpha}$ moves to machine 1, there must be some jobs (i.e., $\gamma$) originally on machine 2 also move to machine 1, increasing the load of machine 1 to a value higher than $l_2(D',s')$.
		Moreover, the incentive of the jobs in $\gamma$ moving to machine 1 is due to the increase of the load of machine 2 by some jobs (i.e., $\beta$) originally on machine 1 moving to machine 2 when $J_{\alpha}$ moves to machine 1.
		In the following, we will show that the jobs in $\beta$ have no incentive to move to machine 2 if inequality \eqref{eq:2unrelated.contradiction} holds, which gives a contradiction.

		By definition, the loads of machine 1 and 2 in $\tilde{s}'$ are
		\begin{gather}
			\label{eq:l1}
			l_1(D',\tilde{s}') = l_1(D',s') + \alpha_1 - \beta_1  + \gamma_1 \,,\\
			\label{eq:l2}
			l_2(D',\tilde{s}') = l_2(D',s') - \alpha_2 + \beta_2 - \gamma_2 \,.
		\end{gather}
		Since schedule $s'$ is a equilibrium, it holds that the cost of job $J_{\alpha}$ in $s'$ is no greater than that in $\tilde{s}'$, that is, 
		\begin{equation}\label{eq:2unrelated.alphacost}
			l_2(D',s') \le l_1(D',\tilde{s}') \,.
		\end{equation}
		First, we know that the job set $\gamma$ is nonempty, otherwise
		\begin{align*}
			l_1(D',\tilde{s}') & = l_1(D',s') + \alpha_1 - \beta_1  + \gamma_1 && \hspace{-2cm} \text{by \eqref{eq:l1}}\\
							   & < l_2(D',s') - \alpha_1  - \gamma_2 - \epsilon_2 + \alpha_1 - \beta_1 + \gamma_1 && \hspace{-2cm} \text{by \eqref{eq:2unrelated.contradiction}}\\
							   & \le l_2(D',s') \,, 
		\end{align*}
		which contradicts with \eqref{eq:2unrelated.alphacost}.

		Now that $\gamma$ is nonempty, in schedule $\tilde{s}'$, the cost of jobs in $\gamma$ is $l_1(D',\tilde{s}')$.
		The reason why jobs in $\gamma$ move to machine 1 when job $J_{\alpha}$ moves to machine 1 is that if any job in $\gamma$ stays at machine 2, the cost will be higher than $l_1(D',\tilde{s}')$.
		Thus we have $l_2(D',\tilde{s}') + \gamma_2 \ge l_1(D',\tilde{s}')$.
		Since $l_1(D',\tilde{s}') \ge l_2(D',s')$ (inequality \eqref{eq:2unrelated.alphacost}), it follows that $l_2(D',\tilde{s}') + \gamma_2 \ge l_2(D',s')$.
		Together with \eqref{eq:l2} we get $\beta_2 \ge \alpha_2$.

		The cost of jobs in $\beta$ in $\tilde{s}'$ is $l_2(D',\tilde{s}')$.
		However, we notice that if jobs in $\beta$ choose machine 1 (after job $J_{\alpha}$ chooses machine 1), the cost of jobs in $\beta$ is at most $l_1(D',s') + \alpha_1$ (since jobs in $\gamma$ will choose machine 2 in this case), and the cost $l_1(D',s') + \alpha_1$ is smaller than $l_2(D',\tilde{s}')$ because
		\begin{align*}
			l_1(D',s') + \alpha_1 & < l_2(D',s') - \gamma_2 - \epsilon_2 && \hspace{-2cm}\text{by \eqref{eq:2unrelated.contradiction}}\\
								  & \le l_2(D',s') - \gamma_2 \\
								  & \le l_2(D',s') - \alpha_2 + \beta_2 - \gamma_2 && \hspace{-2cm}\text{by $\beta_2 \ge \alpha_2$} \\
								  & = l_2(D',\tilde{s}') && \hspace{-2cm}\text{by \eqref{eq:l2}} \,.
		\end{align*}	
		Therefore, the jobs in $\beta$ have no incentive to choose machine 2 in $\tilde{s}'$, since the cost of choosing machine 1 is lower. 
		In other words, schedule $\tilde{s}'$ is not a equilibrium, which is a contradiction.
		Thus, we conclude that $l_1(D',s') \ge l_2(D',s') - \alpha_1  - \gamma_2 - \epsilon_2$, which proves this claim.
		\qed
	\end{proof}

	From \eqref{eq:2unrelated.main} and the above claim, we have
	\begin{align}
		SPE_{D} \left( P_{[\ell:n]} \right) & \le l_1(D',s') + \sum_{j = \ell+1}^n \min_i p_{ij} \nonumber\\
											\label{eq:2unrelated.D'}
											& \le ||D'||_{\infty} + \Delta SPE \left( P_{[\ell+1:n]} \right) + \sum_{j = \ell+1}^n \min_i p_{ij}  \,.
	\end{align}
	Moreover, the cost of job $J_{\ell}$, namely $l_1(D',s')$, must be no greater than the cost of choosing machine 2:
	\[
		l_1(D',s') \le l_2(D'',s'') \,,
	\]
	where $D'' = (d_1, d_2 + p_{2\ell})$ is the initial load if job $J_{\ell}$ chooses machine 2, and $s''$ is the schedule of the jobs $\{J_{\ell +1} ,J_{\ell +2}, \ldots,J_n\}$ playing with initial load $D''$.
	Thus we obtain
	\begin{align}
		SPE_{D} \left( P_{[\ell:n]} \right) & \le l_1(D',s') + \sum_{j = \ell+1}^n \min_i p_{ij} \nonumber\\
											& \le l_2(D'',s'') + \sum_{j = \ell+1}^n \min_i p_{ij} \nonumber\\
											\label{eq:2unrelated.D''}
											& \le ||D''||_{\infty} + \Delta SPE \left( P_{[\ell+1:n]} \right) + \sum_{j = \ell+1}^n \min_i p_{ij} \,.
	\end{align}
	
	According to \eqref{eq:2unrelated.D'} and \eqref{eq:2unrelated.D''}, it holds that
	\[
		SPE_{D} \left( P_{[\ell:n]} \right) \le \min \left\{ ||D'||_{\infty},||D''||_{\infty} \right\} + \Delta SPE \left( P_{[\ell+1:n]} \right) + \sum_{j = \ell+1}^n \min_i p_{ij} \,.
	\]
	As $D' = (d_1 + p_{1\ell}, d_2)$ and $D'' = (d_1, d_2 + p_{2\ell})$, we know
	\[
		\min \left\{ ||D'||_{\infty},||D''||_{\infty} \right\} \le \max \left\{ d_1, d_2 \right\} + \min \left\{ p_{1\ell}, p_{2\ell} \right\} \,.
	\]
	Thus it follows that
	\begin{align*}
		SPE_{D} \left( P_{[\ell:n]} \right) & \le \max \left\{ d_1, d_2 \right\} + \min \left\{ p_{1\ell}, p_{2\ell} \right\} + \Delta SPE \left( P_{[\ell+1:n]} \right) + \sum_{j = \ell+1}^n \min_i p_{ij} \\
											& = ||D||_{\infty} + \Delta SPE \left( P_{[\ell+1:n]} \right) + \sum_{j = \ell}^n \min_i p_{ij} \,,
	\end{align*}
	that is, 
	\begin{equation}\label{eq:2unrelated.case2}
		SPE_{D} \left( P_{[\ell:n]} \right) - ||D||_{\infty} \le \Delta SPE \left( P_{[\ell+1:n]} \right) + \sum_{j = \ell}^n \min_i p_{ij} \,.
	\end{equation}

	Since inequality \eqref{eq:2unrelated.case1} holds for the case $d_1 + p_{1\ell} < d_2$, and inequality \eqref{eq:2unrelated.case2} holds for the case $d_1 + p_{1\ell} \ge d_2$, we obtain that inequality \eqref{eq:2unrelated.case2} holds for any $D$, that is,
	$\Delta SPE \left( P_{[\ell:n]} \right) \le \Delta SPE \left( P_{[\ell+1:n]} \right) + \sum_{j = \ell}^n \min_i p_{ij}$, which concludes the proof of the lemma.
	\qed
\end{proof}

\begin{theorem}\label{thm:UB4SPoA}
	For two machines, the \SPoA is at most $2(n-1)$.
\end{theorem}
\begin{proof}
	Applying Lemma~\ref{lem:2unrelated.key}, we have
	\begin{align*}
		\Delta SPE \left( P_{[1:n]} \right) & \le \Delta SPE \left( P_{[2:n]} \right) + \sum_{j = 1}^n \min_i p_{ij} \\
											& \le \Delta SPE \left( P_{[3:n]} \right) + 2\sum_{j = 1}^n \min_i p_{ij} \\
											& \le \ldots \\
											& \le (n-1) \sum_{j = 1}^n \min_i p_{ij} \,.
	\end{align*}
	Since the optimal cost is at least $OPT \ge \sum_{j = 1}^n \min_i p_{ij}/{2}$ (for 2 machines), it follows that
	\[
		\SPoA \le \frac{\Delta SPE \left( P_{[1:n]} \right)}{OPT} \le 2(n-1) \,,
	\]
	which completes the proof.
	\qed
\end{proof}


\section{Linear upper bound on the \SPoS}
\label{sec:linear_upper_bound}

In this section, we give a \emph{linear upper bound} on the  sequential price of stability for two machines (\Cref{th:SPoS} below). Unlike in the case of the sequential price of anarchy, here we have the freedom to  choose the order of the players.  Each player can choose \emph{any}  tie-breaking rule. Since we consider a full information setting, the tie-breaking rules are also public knowledge.

Though finding the best order can be difficult, we found that a large set of permutations already gives a linear upper bound on \SPoS. In particular, it is enough that the authority divides the players into \emph{two groups} and puts players in the first group first, followed by the players from the second group. Inside each group players can form \emph{any order}. The main result of this section is the following theorem:

\begin{theorem}\label{th:SPoS}
For two machines, the \SPoS is at most $\frac{n}{2}+1$.
\end{theorem}

The proof is in the appendix.
This result cannot be extended to \emph{three} or more machines, because the third machine changes the logic of the proof. In particular, we can no longer assume that the players on the second machine in the optimal assignment can guarantee low costs for themselves  by simply staying on that machine.
For two machines, we conjecture that actually there is always an order which leads to the optimum: 

\begin{conjecture}
     For two machines, the {\bf SPoS} is $1$.
\end{conjecture}

Though we are not able to prove this conjecture, in the next section, we introduce a more restricted solution concept, and show that in that case the optimum can be achieved. 

\section{Achieving the optimum: the \aSPoS}
\label{sec:optimum_sequence}

In this section, we study the adaptive sequential price of stability. Unlike the previous models, here we assume that there is some authority, which has full control over the order of the players' arrival in the game. It does not only fix the initial complete order, but can also change the order of arrivals depending on the decision that previous players made. 
On the other hand, the players still have the freedom to choose any action in a given state, each of them aiming at  minimizing her own final cost.
The players also know the whole decision tree, and thus the way the authority chooses the order. As in the previous section, each player can use \emph{any} tie-breaking rule, and the tie-breaking rules are also known to all players.

This model is the closest instantiation of a general extensive form game compared to the previously studied models in this paper. 
In this way, the authority has an option to punish players for deviating from the optimal path (path leading to a social optimum) by placing different players after the deviating decisions of the deviating player. As a result, rational players may achieve much better solutions in the end. The following theorem shows that achieving the optimum solution is possible for $2$ machines:


\begin{theorem}\label{optimum_sequence}
For two machines, the $\aSPoS$ is $1$.
\end{theorem}


The proof is in the appendix.
The previous result cannot be extended to more than $2$ machines:

\begin{theorem}\label{th:adaptive:LB}
For three or more machines, the \aSPoS is at least $\frac{3}{2}$.
\end{theorem}

\begin{proof}
%
Consider the following instance with three machines and three jobs, where the optimum is shown as gray boxes:
\begin{align*}\begin{array}{c||c|c|c|}
&  J_1 & J_2 & J_3 \\ \hline \hline
M_1 & 4 -\epsilon &  \take{2} & \take{2}  \\ \hline
M_2 & \take{4} & 3 & 3  \\ \hline
M_3 & 6 & 6-\epsilon & 6-\epsilon  \\ \hline 
\end{array} &&
\end{align*}
We distinguish two cases for the first player to move (the root of the tree), and show that in neither case the players will implement the optimum:
\begin{enumerate}
	\item \emph{The first to move is $J_1$.} This player will choose the cheapest machine $M_1$, because none will join this machine. Indeed, the second player to move will choose $M_2$ knowing that the last one will then choose $M_3$. 
	\item \emph{The first to move is $J_2$ or $J_3$.} This player will choose $M_2$ and \emph{not} $M_1$. Indeed, if the first player to move, say $J_2$, chooses $M_1$, then either (I) the other two follow also the optimum (which costs $4$ to $J_2$) or (II) they choose another solution,  whose cost is at least $6-\epsilon$. In the latter case, we have the lower bound. In case (I), we argue that choosing $M_2$ is better for $J_2$, because no other player will join: for the following players,  being both on machine $M_1$ is already cheaper than being on $M_2$ with $J_2$.
\end{enumerate}
In the first case, given that $J_1$ is allocated to $M_1$, the cheapest solution costs $6-\epsilon$. In the second case, one among $J_2$ or $J_3$ is allocated to $M_2$. The best solution, in this case, costs again $6-\epsilon$. This completes the proof.
\qed\end{proof}

\begin{remark}\label{rem:thre-machines}
The following example shows that the analysis of \Cref{optimum_sequence} cannot be extended to $3$ machines even in the case of identical machines. Assume that we have $m=3$ machines, the initial loads on these machines are $(0,2,6)$ and there are $3$ jobs left to be assigned with processing times $7, 5$ and $5$. Note that the constrained optimum here is $(10, 9, 6)$, that is the first job with processing time $7$ gets assigned to the second machine $M_2$, while both jobs with processing times $5$ and $5$ get assigned to machine $M_1$. On the other hand, if any of these players chooses different machine their cost is strictly decreasing in the subgame perfect equilibrium solution. We did not find any example showing that  $\aSPoS> 1$ for more than $2$ identical machines, unlike the case of unrelated machines.      

\end{remark}

\section{Conclusions}\label{sec:conclusion}

In this paper, we disprove a conjecture from \cite{Hassin} and give a linear lower bound construction for the sequential price of anarchy. On the other hand, we show linear upper bound. For the best sequence of players, we prove a linear upper bound, that is $4$ times lower than the upper bound for sequential price of anarchy. Moreover, we prove the existence of a sequential extensive game which gives an optimum solution. One possible direction for future research is to investigate whether the sequential price of stability is $1$ for any number of \emph{identical} machines. In this work, we give some evidence that the case of three (or more) machines is different from the case of two machines (see \Cref{th:adaptive:LB}
and \Cref{rem:thre-machines}).

Our linear lower bound on the sequential price of anarchy (\Cref{counterexample}) suggests that subgame perfect equilibria do not guarantee in the worst case a price of anarchy independent of the number of jobs, even for two machines. Though our lower bound is based on a suitable tie-breaking rule, we believe it holds without any tie being involved (Conjecture~\ref{conj:SPOA}).

\paragraph{Acknowledgments.} We are grateful to Thomas Erlebach for spotting a mistake in an earlier proof of Theorem~\ref{optimum_sequence} and for suggesting a fix of the proof. We thank Paul D\"{u}tting for valuable discussions. We also thank anonymous reviewers and seminar participants of OR 2016 for suggestions that improved the paper.  

\bibliographystyle{plain}
\bibliography{sequential,Mybib}

\newpage

\section{Appendix}

\begin{proof}[of Theorem~\ref{th:LB:five-jobs}]
	Consider the following instance with jobs arriving in this order (from left to right), 
	\begin{align} \label{eq:SPE:five-players}
	\begin{array}{c||c| c| c | c | c|}
	& J_1 & J_2 & J_3 & J_4 & J_5 \\ \hline\hline
	M_1 & \take{3-11\epsilon} & \epsilon & \take{\epsilon} & 1-2\epsilon & 2-8\epsilon \\ \hline
	M_2 & \epsilon & \take{2-9\epsilon} & 2-8\epsilon & \take{1-2\epsilon} & \take{1-2\epsilon} \\ \hline
	\end{array}
	\end{align}
	where the subgame perfect equilibrium is shown as  gray boxes.
	Note that the optimum has cost $1$, while the subgame perfect equilibrium costs $4-13\epsilon$. By letting $\epsilon$ tend to $0$ we get the desired result.  
	
	For ease of presentation, we assume that \emph{players break ties in favor of machine $M_1$} (this assumption can be dropped by using more involved coefficients for the $\epsilon$ terms -- see Remark~\ref{rem:tie-breaking}). To see why the allocation corresponding to the gray boxes is indeed a subgame perfect equilibrium,  we note the following:
	\begin{enumerate}
		\item  If the first three jobs follow the optimum, then $J_4$ prefers to deviate to $M_2$, which causes $J_5$ to switch to machine $M_1$:
		\begin{align*}
		\begin{array}{c||c|c|c|}
		& J_1 & J_2 & J_3  \\ \hline \hline
		M_1 & 3-11\epsilon & \take{\epsilon} & \take{\epsilon}  \\ \hline 
		M_2 & \take{\epsilon} & 2-9\epsilon & 2-8\epsilon \\ \hline
		\end{array} &&
		\Rightarrow &&
		\begin{array}{|c|c|}
		J_4 & J_5 \\ \hline\hline
		1-2\epsilon & \take{2-8\epsilon} \\ \hline
		\take{1-2\epsilon} & 1-2\epsilon \\ \hline
		\end{array}
		\end{align*} Now the cost for $J_3$ would be $2-6\epsilon$.
		\item  Given the previous situation, $J_3$ prefers to deviate to $M_2$ because in this way $J_4$ and $J_5$ choose $M_1$, and her cost will be $2-7\epsilon$:
		\begin{align*}\begin{array}{c||c|c|}
		& J_1 & J_2  \\ \hline \hline
		M_1 & 3-11\epsilon & \take{\epsilon}   \\ \hline 
		M_2 & \take{\epsilon} & 2-9\epsilon  \\ \hline
		\end{array} &&
		\Rightarrow &&
		\begin{array}{|c|c|c|}
		J_3 &	J_4 & J_5 \\ \hline\hline
		\epsilon & \take{1-2\epsilon }& \take{2-8\epsilon} \\ \hline
		\take{2-8\epsilon} & 1-2\epsilon & 1-2\epsilon \\ \hline
		\end{array}
		\end{align*}
		Now the cost for $J_2$ would be $3 - 9\epsilon$.
		\item  Given the previous situation, $J_2$ prefers to deviate to  $M_2$ because in this way $J_3$ and $J_4$ choose $M_1$,  $J_5$ chooses $M_2$, and the cost for $J_2$ is $3-10\epsilon$:
		\begin{align*}\begin{array}{c||c|}
		& J_1   \\ \hline \hline
		M_1 & 3-11\epsilon  \\ \hline 
		M_2 & \take{\epsilon}  \\ \hline
		\end{array} &&
		\Rightarrow &&
		\begin{array}{|c|c|c|c|}
		J_2 & J_3 &	J_4 & J_5 \\ \hline\hline
		\epsilon   & \take{\epsilon} & \take{1-2\epsilon }& 2-8\epsilon \\ \hline
		\take{2-9\epsilon}  & 2-8\epsilon & 1-2\epsilon & \take{1-2\epsilon} \\ \hline
		\end{array}
		\end{align*} 
		Now the cost for $J_1$ would be $3 - 10\epsilon$.
	\end{enumerate}
	We have thus shown that, if $J_1$ chooses $M_2$ then her cost will be $3 - 10\epsilon$.  To conclude the proof, observe that if $J_1$ chooses $M_1$, then by similar arguments as above job $J_2$ prefers to choose machine $M_2$ and all players will choose the allocation in \eqref{eq:SPE:five-players}, the cost for $J_1$, in this case, is also $3  - 10 \epsilon$. Since players break ties in favor of $M_1$, we conclude that the  subgame perfect equilibrium is the one in \eqref{eq:SPE:five-players}.
\qed\end{proof}
	

\begin{remark}[tie-breaking rule]\label{rem:tie-breaking}
	In the construction above, we have used the fact that players break ties in favor of $M_1$. This assumption can be removed by using slightly more involved coefficients for the $\epsilon$ terms, so that ties never occur.
\end{remark}

\begin{proof}[of Theorem~\ref{th:SPoS}]
Consider an optimal assignment and denote the corresponding makespan by $OPT$. By renaming jobs and machines, we can assume without loss of generality that in this optimal assignment machine $M_1$ gets the first $k\leq \frac{n}{2}$ jobs, and machine $M_2$ gets all the other jobs:
\begin{align*}
\{J_1,J_2,\ldots, J_k\} \rightarrow M_1 \enspace,  && \{J_{k+1}, \ldots, J_n\} \rightarrow M_2 \enspace .
\end{align*}
Take the sequence given by the jobs allocated to $M_1$ followed by the jobs allocated to $M_2$, 
$$J_1, J_2, \ldots, J_k, J_{k+1},\ldots, J_n.$$ We prove that for this sequence there is a subgame perfect equilibrium whose makespan is at most $(k+1)\cdot OPT$. 

In the proof, we consider the \emph{first player who deviates} from the optimal allocation. We distinguish two cases, corresponding to the next two claims. 

\begin{claim}
	If the first player $J_d$ who deviates is in $\{J_{k+1}, \ldots, J_n\}$, then she does not improve her own cost. 
\end{claim}
\begin{proof}[of Claim]
	Observe that all players in $\{J_1,J_2,\ldots, J_k\}$, who came before player $J_d$, did not deviate. Machine $M_1$ has thus exactly the jobs that it gets in the optimum. If $J_d$ stays on $M_2$, her cost will be at most $OPT$ (in the worst-case, all subsequent jobs choose to stay on $M_2$). Moving to $M_1$ will in the end produce a schedule with fewer jobs on $M_2$ and more jobs on $M_1$, compared to the optimum. The cost on $M_1$ is therefore at least $OPT$ (otherwise the new schedule has smaller makespan than $OPT$, a contradiction with the optimality), which is the cost of $J_d$ when deviating. 
\end{proof}

The remaining  case is the following one.

\begin{claim}
	If the first player $J_d$ who deviates is in $\{J_{1}, \ldots, J_k\}$, then any subgame perfect equilibrium has  makespan at most $(t+1)\cdot OPT$ where $t=k+1-d$.
\end{claim} 
\begin{proof}[of Claim]
	The proof is by induction on $t$. For $t=1$, the deviating player is the last in $\{J_{1}, \ldots, J_k\}$, i.e., $J_k$. 
	Note that, if $J_k$ does not deviate, then by the previous claim, $J_k$ guarantees her cost to be at most $OPT$.
	Thus, if $J_k$ deviates to $M_2$, then in the resulting equilibrium schedule she cannot pay more, i.e., she pays at most $OPT$.
	We now argue that if $J_k$ deviates, $M_1$ will have, in the resulting equilibrium, load at most $2\cdot OPT$.
	Clearly, in the resulting equilibrium, the load on $M_2$ is at most $OPT$. Some of the jobs among $J_{k+1},\ldots,J_n$ may be assigned to $M_1$ in this equilibrium. Moving them all to machine $M_2$ will result in a load on $M_2$ being at most $2\cdot OPT$ (since all these jobs are, in the optimum solution, on $M_2$). 
	Hence, each job among $J_{k+1},\ldots,J_n$ that decided to move to $M_1$ in the resulting equilibrium cannot have a worse cost than that of staying on $M_2$, which guarantees cost at most $2\cdot OPT$.

	For $t>0$, the first player who deviates from the optimal assignment is $J_{k-t}$.  We argue similarly that the makespan is at most $(t+1) \cdot OPT$. By induction we can assume that if player $J_{k-t}$ stays on the first machine then she is guaranteed to pay at most $t\cdot OPT$. Since in the subgame perfect equilibrium $J_{k-t}$ chooses the second machine, we know that she is paying at most $t\cdot OPT$ on the second machine. Thus, the cost on the second machine is at most $t\cdot OPT$. We next argue that the cost on the first machine is at most $(t+1)\cdot OPT$. If no player in $\{J_{k+1}, \ldots, J_{n}\}$ chooses the first machine, then the cost of this machine is at most $OPT$. Otherwise, if some player $J''$ in $\{J_{k+1}, \ldots, J_{n}\}$ chooses the first machine, then we show that she is paying at most $(t+1) \cdot OPT$, thus implying that the cost of the first machine is at most $(t+1) \cdot OPT$.  This is because, if  $J''$ would choose the second machine, she would pay at most $(t+1) \cdot OPT$: indeed, the cost on the second machine was at most $t \cdot OPT$ and even if all players after $J''$ choose $M_2$, they will contribute at most another additional factor $OPT$ (because the players after $J''$ are in $\{J_{k+1}, \ldots, J_{n}\}$). 
	
	We have thus shown that the cost on the first machine is at most $(t+1)\cdot OPT$, and therefore this sequence results in a solution which has  makespan at most $(k+1)\cdot OPT$, which completes the proof.	
\qed\end{proof}

The two claims above imply the theorem.\qed\end{proof}

\begin{proof}[of Theorem~\ref{optimum_sequence}][Main Idea]
	The main idea of the proof is as follows. 
	Each internal node of a tree corresponds to a choice of some player, and the path (edges) to that node correspond to an allocation of a \emph{subset of players} (the nodes on the node-to-root path). We consider the corresponding \emph{constrained optimum}, that is, allocation of all remaining jobs that minimizes the makespan, given the fixed allocation of the previous players. Among these remaining players, we then find a particular one for which the constrained optimum is better than any constrained optimum if she deviates. If this player deviates, we can \emph{punish} such deviation by letting the others implement the more expensive constrained optimum (by adaptively fixing their order). 
	
\end{proof}

\begin{proof}
	\newcommand{\copt}{opt}
	For any node $h$ of the tree, let $S_h$ be the subset of players that appeared on the previous nodes (i.e., from the parent of $h$ up to the root), and let $A_h$ be the  resulting allocation (described by the path). Let $\copt_h$ be the constrained optimum, that is, the allocation of all remaining jobs $R_h$ which, combined with $A_h$, minimizes the resulting makespan. We now choose a suitable player $J^*(h)$ to put on node $h$, according to the following:
	
	\begin{claim}
		There exists a player $J^*(h)$ among the remaining players $R_h$ such that the following holds. If $J^*(h)$ deviates from the constrained optimum $\copt_h$, then the new constrained optimum (if implemented) is not less costly for $J^*(h)$.  
	\end{claim}
		\begin{proof}[of Claim]
	Consider the best alternative solution $\copt'_h$, that is, the best allocation for the remaining jobs $R_h$, given the allocation $A_h$ of jobs in $S_h$, such that at least one job in $R_h$ is not allocated as in $\copt_h$. 
	Observe the following: 
	\begin{itemize}
		\item There must be a job $J'\in R_h$ which $\copt'_h$ allocates differently from $\copt_h$ and in $\copt'_h$ this job is on a machine determining the makespan of $\copt'_h$. 

	\end{itemize}
	This means that if $J'$ deviates from $\copt_h$, then she is choosing the machine according to $\copt'_h$. Consider the resulting constraints $A'_h$ and $R'_h$, where we extend $A_h$ by allocating $J'$ as in $\copt'_h$ (therefore the remaining players are $R'_h=R_h \setminus \{J'\}$). Since $\copt'_h$ is the best solution among all solutions in which at least one job from $R_h$ is not allocated as in $\copt_h$, then $\copt'_h$ is a constrained optimum for such $A'_h$ and $R'_h$.  The cost of $J'$ in solution $\copt'_h$ is the corresponding makespan, which cannot be smaller than the makespan of $\copt_h$.  We can thus choose $J^*(h)=J'$ and observe that  deviating from $\copt_h$ makes $J^*(h)$ incur at least the same cost as in $\copt_h$. \qed
\end{proof}
   At each node $h$, the chosen player $J^*(h)$ can either follow the constrained optimum, or deviate. backward induction guarantees that in either case, the remaining players implement the resulting constrained optimum. The above claim implies that $J^*(h)$ does not deviate from the constrained optimum $\copt_h$ (under a suitable tie-breaking rule).\qed
\end{proof}

\end{document}